\documentclass[twocolumn]{autart}    

\usepackage{graphicx}
\usepackage{amsmath,amssymb,amsfonts}
\newenvironment{proof}{\noindent\textit{Proof.}}{\hfill$\blacksquare$}
\usepackage{cite}
\usepackage[normalem]{ulem}
\usepackage{mathtools}
\newtheorem{propos}{Proposition}  
\usepackage{xcolor}
\usepackage{booktabs}
\usepackage{float}
\usepackage{comment}

\begin{document}

\begin{frontmatter}

\title{Adaptive MPC for Constrained Trajectory Tracking of Uncertain LTI System with Input-Rate Limits} 

\thanks{Corresponding author B.~Dey}

\author[Paestum]{Bishal Dey}\ead{dey.bishal@gmail.com},    
\author[Rome]{Abhishek Dhar}\ead{abhishek.dharr@gmail.com},               
\author[Baiae]{Sumit Kr. Pandey}\ead{skpdmk@gmail.com},  
\author[Paestum]{Anindita Sengupta}\ead{aninsen@ee.iiests.ac.in} 
\address[Paestum]{Department of Electrical Engineering, Indian Institute of Engineering Science and Technology, Howrah}  
\address[Rome]{Division of Technology and Innovation, Epiroc, Sweden}             
\address[Baiae]{Faculty of Science and Engineering, Jharkhand Rai University, Ranchi}        


\begin{abstract}                          
This paper addresses the trajectory-tracking problem for discrete-time linear time-invariant systems with bounded parametric uncertainty, subject to hard constraints on system states, control inputs, and input rates. Unlike existing methods, which often consider only partial uncertainty, omit input-rate or state constraints, or focus on regulation problems, this work provides a systematic adaptive model predictive control (MPC) solution for constrained trajectory tracking under full parametric uncertainty. Determining the control input required to achieve zero tracking error under unknown parameters is challenging. Simultaneously, trajectory tracking under uncertainty with input-rate constraints induces temporal coupling in the control sequence, resulting in a time-varying admissible control set and rendering standard recursive feasibility arguments inapplicable. These challenges are overcome by systematically utilizing the estimated system parameters, coupled with a suitably designed adaptive learning process within a reformulated MPC framework. The recursive feasibility of the proposed MPC optimization routine is then rigorously established despite the time-varying admissible control set induced by input-rate constraints. Closed-loop stability is guaranteed via Lyapunov-based analysis, ensuring convergence of the tracking error and boundedness of system states. Simulation results validate the effectiveness of the proposed approach in achieving robust constrained trajectory tracking under uncertainty.
\end{abstract}

\end{frontmatter}

\section{Introduction}
The design of control frameworks for trajectory tracking in constrained dynamical systems with bounded plant-parameter uncertainty is critical in safety-critical applications such as aerospace systems \cite{Tian2023}, autonomous vehicles \cite{Hu2024, Dey2025}, and robotics \cite{Liu2022}. In such systems, controllers must ensure accurate tracking while respecting constraints arising from actuator limitations, safety requirements on system states, and restrictions on control input rates. In practice, plant-parameter uncertainty is unavoidable due to modelling inaccuracies, unmodeled dynamics, environmental variations, component ageing, manufacturing tolerances, or time-varying characteristics, leading to deviations between the nominal model and actual system behaviour and potentially causing performance degradation or constraint violations. Consequently, developing control strategies that ensure robust trajectory tracking under multiple constraints in the presence of uncertainty remains a significant challenge.
\par This challenge is further intensified when hard constraints are imposed simultaneously on system states, control inputs, and input increments. State constraints ensure safe operation within physical limits, while input constraints arise from actuator saturation or hardware limitations. Input-rate constraints are imposed to prevent abrupt control variations, reduce actuator wear, avoid excitation of unmodeled high-frequency dynamics, and ensure smooth implementation. However, the input-rate constraints introduce additional complexity by coupling successive control inputs and restricting feasible control sequences, thereby reducing the admissible solution space. In the presence of plant-parameter uncertainty, this coupling further complicates trajectory prediction and makes it difficult to guarantee both constraint satisfaction and tracking performance. Existing studies have primarily addressed either constrained regulation problems for uncertain systems \cite{Ghosh2026}, \cite{Dhar2021} or unconstrained \cite{Yang2021} or partially constrained \cite{Kohler2026} trajectory tracking under uncertainty.
\par Several control strategies have been proposed to address constrained tracking problems. 
Among these methods, model predictive control (MPC) has become one of the most widely used frameworks for constrained systems \cite{Kouvaritakis2016}, as it explicitly incorporates state and input constraints within an online optimization routine. However, classical MPC relies on accurate plant models, which are rarely available in practice due to modelling errors and parameter uncertainty, potentially resulting in infeasible or degraded solutions. Therefore, various extensions of MPC have been developed to address this problem. The tube-based MPC (TMPC) strategies \cite{Rakovic2012}, \cite{Langson2003} provide robustness against additive uncertainties by confining system trajectories within invariant tubes around nominal trajectories. Recent studies have also explored combining MPC with machine learning \cite{yao2023discriminative} and deep learning \cite{hu2023coordinated} for uncertainty modelling \cite{rosolia2020learning}. However, such data-driven approaches typically lack guarantees of recursive feasibility, robustness, and closed-loop stability. Adaptive MPC approaches \cite{Tanaskovic2014}, \cite{fan2016further} address parametric uncertainty through online parameter estimation. Early methods required identification of the full set of uncertain parameters \cite{Lorenzen2017}, while later approaches reduced computational complexity by requiring only parameter bounds \cite{Dhar2025} or using gradient-descent-based adaptive laws \cite{zhu2015adaptive}. Other approaches, such as control barrier functions \cite{ames2014control}, reference governors \cite{bemporad1997nonlinear}, anti-windup compensation \cite{lavretsky2007stable}, and saturation-based designs \cite{niu2020global}, have also been proposed for enforcing safety and input constraints. 
\par Despite these advances, extending existing approaches to trajectory-tracking problems under uncertainty remains challenging, particularly when input-rate constraints are imposed alongside state and input constraints. Unlike regulation problems, trajectory tracking requires continuous adaptation of control actions to follow a time-varying reference. In such settings, input-rate constraints are not only practical but often necessary to avoid abrupt variations in control signals required to drive the system along the desired trajectory. This ensures smooth actuator operation and prevents excessive wear or excitation of undesirable dynamics. However, the presence of input-rate constraints further limits control flexibility by restricting how rapidly control inputs can evolve over time, thereby limiting the system’s ability to respond to both reference variations and uncertainty. In uncertain systems, this poses a fundamental difficulty in constructing control sequences that simultaneously ensure accurate tracking, satisfy all constraints at every time step, and remain robust to parameter variations. As a result, guaranteeing constraint satisfaction and performance across all admissible uncertainty realizations becomes significantly more complex. This motivates the development of new control frameworks that address these challenges systematically and computationally tractably.
\par In this article, an adaptive MPC framework is proposed for trajectory tracking of a discrete-time linear time-invariant (LTI) system with bounded parametric uncertainty, subject to hard constraints on system states, control inputs, and input rates. The system dynamics are expressed in equivalent error coordinates with respect to a time-varying reference, thereby transforming the trajectory-tracking problem into a regulation problem. The system constraints, including state, input, and input-rate constraints, are meticulously reformulated in the error-dynamics domain to ensure consistency with the original system's physical constraints and to facilitate their enforcement in the control design. A key challenge in solving the tracking problem in the presence of parametric uncertainty is determining the appropriate input required to maintain zero tracking error, as such an input generally depends on the exact knowledge of the system parameters. This challenge is systematically addressed by constructing the required input signal using adaptively estimated parameters, which are subsequently incorporated into the control design to guarantee the desired system-theoretic properties. The proposed control framework ensures recursive feasibility, guaranteeing that the optimization problem remains solvable at all subsequent time steps. Furthermore, closed-loop stability of the system is analyzed to demonstrate convergence of the error dynamics to the origin while ensuring boundedness of the physical system states.
\par\textit{Notations:} If a nonempty set $\mathcal{P}\subset\mathbb{R}^n$ is compact and convex, then it is termed a $\mathcal{C}$-set, and if in its interior the origin lies, then it is called a $\mathcal{C}_0$-set. The diameter of $\mathcal{P}$ is $\mathrm{diam}(\mathcal{P})=\max\{\|m-n\|_F:m,n\in\mathcal{P}\}$. The symbols $\|\cdot\|$, $\|\cdot\|_\infty$, and $\|\cdot\|_F$ denote the Euclidean, infinity, and Frobenius norms, respectively. For two sets $\mathcal{P}$ and $\mathcal{Q}$, the Minkowski sum and Pontryagin difference are defined as $\mathcal{P}\oplus\mathcal{Q}=\{a+b:a\in\mathcal{P},\,b\in\mathcal{Q}\}$ and $\mathcal{P}\ominus\mathcal{Q}=\{a:a+\mathcal{Q}\subseteq\mathcal{P}\}$. The index sets are defined as $\mathbb{I}=\{0,1,2,\ldots\}$, $\mathbb{I}_N=\{0,1,\ldots,N\}$, $\mathbb{I}^+=\{1,2,\ldots\}$, and $\mathbb{I}_N^+=\{1,2,\ldots,N\}$. For $P\in\mathbb{R}^{n\times n}$ and $x\in\mathbb{R}^n$, the weighted quadratic norm is $\|x\|_P^2=x^\top P x$. The convex hull generated by points $\{s_1,\ldots,s_m\}\subset\mathbb{R}^n$ is denoted by $\mathrm{co}\{s_1,\ldots,s_m\}$. A compact ball centered at $M\in\mathbb{R}^{n\times m}$ with radius $r$ is given by $B_r(M)=\{M' : \|M-M'\|_F\le r\}$.

\section{Problem Statement}
 In this work, a constrained uncertain discrete-time LTI plant is taken into consideration
 \begin{equation}
    x_{k+1}=Ax_k+Bu_k= \Theta X_k 
    \label{eq:1}
\end{equation}
\begin{subequations}
\begin{align}
    &x_k \in \mathcal{X}
    \label{eq:StateC}\\
    &u_k \in \mathcal{U}
    \label{eq:IC}
     \\
    \Delta &u_k \in \mathcal{U}_\Delta \label{eq:RIC} 
\end{align}
\end{subequations}
where, $\Delta u_k=u_k-u_{k-1}$. The lumped parameter $\Theta =[A,B]$ is uncertain and $X_k=[x_k^T,u_k^T]^T$ is the lumped regressor vector. $A \in \mathbb{R}^{n \times n}$ and $B \in \mathbb{R}^{n \times m}$ are unknown and constant matrices. The constraints $\mathcal{X} \subset \mathbb{R}^n$, $\mathcal{U} \subset \mathbb{R}^m$ and $\mathcal{U}_\Delta\subset \mathbb{R}^m$ are $C_0-$sets.
\begin{assum}\label{assum:lumped}
    The lumped uncertain parameter $\Theta$ satisfies the following
\begin{equation}
    \Theta \in \Psi \triangleq \textnormal{co}\{\psi_1,\psi_2,\cdots,\psi_G\}\nonumber
\end{equation}
where $\psi_i\triangleq [A_i,B_i],\forall i\in \mathbb{I}^+_G$ are known.
\end{assum}
The objective is to design a control strategy that enables the states of the uncertain system \eqref{eq:1} to track a reference trajectory $x^r_k$, which satisfies the following assumption
\begin{assum}\label{assum:ref}
    $x^r_k\in\mathcal{X}^r,\forall k\in \mathbb{I}^+$, where $\mathcal{X}^r\triangleq \text{co}\{x^r_1,x^r_2,\cdots,x^r_L\}\subset \mathbb{R}^n$, where $L\in\mathbb{I}^+$ is finite and $x^r_i,\forall i\in \mathbb{I}_L^+$ are known. 
\end{assum}
\begin{propos}\label{propos:xrbound}
    If $x^r_k$ is reachable for all $[\bar{A},\bar{B}]\in \Psi$, then there exists a $u^r_k\in \mathcal{U}^r$, such that
\begin{equation}
    x^r_{k+1}=Ax^r_k+Bu^r_k
\end{equation}
where $[A,B]$ are the uncertain system parameters and $\mathcal{U}^r=\{u^r_k|u^r_k=(\bar{B}^T\bar{B})^{-1}\bar{B}^T(x^r_{k+1}-\bar{A}x^r_k),[\bar{A},\bar{B}]\in \Psi,x^r_k\in \mathcal{X}^r\}$. 
\end{propos}
\begin{lem}\label{lem:ubound}
    If $u^r_k\in \mathcal{U}^r, \forall k \in \mathbb{I}^+$, then the bound on the input rate $\Delta u^r_k$ is given as
\begin{equation}
    \|\Delta u^r_k\|\leq d_u:=\mathrm{diam}(\mathcal{U}^r),\quad \forall k \in \mathbb{I}^+
    \label{eq:ubound}
\end{equation} 
\end{lem}
\begin{proof} 
Since $u^r_k\in\mathcal{U}^r,\forall k \in \mathbb{I}^+$, it follows, $\|\Delta u^r_k\|=\|u^r_{k}-u^r_{k-1}\|\leq d_u\triangleq \text{diam}(\mathcal{U}^r), \forall k \in \mathbb{I}^+$. Hence, $\Delta u^r_k \in \mathcal{B}_{d_u}(\bar{\textbf{0}}), \forall k\in \mathbb{I}^+$, where $\bar{\textbf{0}}\in \mathbb{R}^m$ is a zero vector. It is concluded that \eqref{eq:ubound} holds. 
\end{proof}
\par \textit{Definition 1:} The $C_0-$set containing the input rate (defined in \textit{Lemma} \ref{lem:ubound}) is defined as $\mathcal{U}^r_\Delta =\{\Delta u^r_k |\|\Delta u^r_k\|\leq d_u, \Delta u^r_k=u^r_k-u^r_{k-1},\forall u^r_k\in \mathbb{R}^m,\forall k\in \mathbb{I}^+\}$. 
\par In case of successful tracking, $x_k\rightarrow x^r_k$, and $u_k\rightarrow u^r_k$. To facilitate the tracking problem, it is reformulated into a regulatory problem as given below
\begin{equation}
    x^e_{k+1}=Ax^e_k+Bu^e_k=\Theta X^e_k
    \label{eq:incremental}
\end{equation}
subjected to
\begin{subequations}
\begin{align}
    &x^e_k\in \mathcal{X}^e:=\{x-x^r|x\in\mathcal{X},x^r\in \mathcal{X}^r\} 
     \label{eq:se}\\
    &u^e_{k} \in \mathcal{U}^e:=\{u-u^r|u \in \mathcal{U},u^r\in \mathcal{U}^r\} \label{eq:ue}
     \\
    \Delta &u^e_{k} \in \mathcal{U}^e_\Delta:=\{\Delta u-\Delta u^r|\Delta u \in \mathcal{U}_\Delta, \Delta u^r \in \mathcal{U}^r_\Delta\} \label{eq:del} 
\end{align}
\end{subequations} 
where 
$X^e_k \triangleq \begin{bmatrix}
    x^{e^T}_{k} ,  u_k^{e^T}
\end{bmatrix}^T \in \mathbb{R}^{(n+m)}$. 
The following COCP presents the regulatory problem \eqref{eq:incremental} in the presence of \eqref{eq:se}, \eqref{eq:ue} and \eqref{eq:del}
\begin{equation}
    \mathbb{P}: \quad \textbf{u}_k^{e^*}=\text{argmin}_{u^e_k}J_k \nonumber
\end{equation}
\begin{subequations}
\begin{align}
    \text{s.t.} \qquad &x^e_{k+i \mid k} 
    = A x^e_{k+i-1 \mid k} + Bu^e_{k+i-1\mid k}, \forall i \in \mathbb{I}_N^+
     \label{eq:5a}\\
    & 
        u^e_{k+i-1\mid k}\in \mathcal{U}^e,
        \Delta u^e_{k+i-1 \mid k}\in \mathcal{U}^e_\Delta
     \\
    &x^e_{k+i \mid k} \in \mathcal{X}^e,x^e_{k+N \mid k} \in \mathcal{X}_{T} \subset \mathcal{X}^e  
\end{align}
\end{subequations} 
 where $x^e_{k+i \mid k}$ and $u^e_{k+i-1 \mid k}$ are the predictions of length N, $J_k$ is the cost function and the solution of the COCP $\mathbb{P}$ is $\textbf{u}_k^{e^*}=\{
    u^{e^*}_{k \mid k}, \cdots, u^{e^*}_{k+N-1 \mid k}
\}^T$. Since $A$ and $B$ are unknown, so \eqref{eq:5a} is unimplementable and hence \eqref{eq:incremental} is reformulated using the following estimated dynamics 
\begin{equation}
    \hat{x}^e_{k+1}=\hat{A}_kx^e_k+\hat{B}_ku^e_k
    \label{eq:estimate}
\end{equation}
 The matrices $\hat{A}_k\in \mathbb{R}^{n\times n}$ and $\hat{B}_k\in \mathbb{R}^{n\times m}$ vary with time, and are updated using a projection-based gradient descent adaptive law. The COCP $\mathbb{P}$ is modified to guarantee trajectory tracking and recursive feasibility of the uncertain plant \eqref{eq:1}.

\section{Indirect Adaptive Law}
This section presents the adaptive law for parameter estimation and the development of error sets that bound the state errors arising from parameter adaptation. 
\subsection{Parameter Update Law}
The estimated system \eqref{eq:estimate} can be expressed in a compact form as follows
\begin{equation}
    \hat{x}^e_{k+1}=\hat{\Theta}_kX^e_k
    \label{eq:compact estimate}
\end{equation}
where $\hat{\Theta}_k\triangleq\begin{bmatrix}
    \hat{A}_k & \hat{B}_k
\end{bmatrix}\in \mathbb{R}^{n \times (n+m)}$. A gradient-descent-based adaptive law, which updates the estimated system parameters are as follows: 
\begin{equation}
    \hat{\Theta}_{k+1}=\hat{\Theta}_k+\lambda \Tilde{x}^e_{k+1}X_k^{e^T}
    \label{eq:15}
\end{equation}
where $\lambda$ is the learning rate and
\begin{equation}
    \Tilde{x}^e_k \triangleq x^e_k-\hat{x}^e_k; \Tilde{x}^e_{k+1}=\Tilde{\Theta}_kX^e_k;\Tilde{\Theta}_k\triangleq\Theta-\hat{\Theta}_k 
\end{equation}
\begin{thm} \label{thm: Adaptive law constraint}
\cite{zhu2015adaptive} For the uncertain system \eqref{eq:incremental}, if there exists a $ u^e_k$ that satisfies 
\begin{equation}
    X_k^{e^T}X^e_k \leq \frac{2-\alpha}{\lambda}
    \label{eq:17}
\end{equation}
given $\alpha \in (0,2)$ and $\lambda >0$, then with the adaptive law \eqref{eq:15}, $\hat{\Theta}_k$ and $ \Tilde{x}^e_k$ is ultimately bounded and $ \Tilde{x}^e_k$ asymptotically converges to zero.
\end{thm}
To ensure $\hat{\Theta}_k\in \Psi$, for all $k \in \mathbb{I}$, the update law \eqref{eq:15} is reformulated as follows  
\begin{equation}
    \hat{\Theta}_{k+1}=\text{Proj}_{\Psi}(\hat{\Theta}_k+\lambda\Tilde{x}^e_{k+1}X_k^{e^T})
    \label{eq:Adaptive}
\end{equation}
where $\Psi$ is known, $\hat{\Theta}_0\in \Psi$ and $\text{Proj}_\Psi(.)$ returns the orthogonal projection of the argument matrix on $\Psi$.
\begin{propos}\label{propos:xm}
    Since the sets $\mathcal{X}^e$ and $\mathcal{U}^e_k$ (Computation of the set $\mathcal{U}^e_k$ is shown in \textit{Section 4.2}) are $C_0$-sets, there exists $x_m$, $ u_m \in \mathbb{R}$ such that $\| x^e_k\| \leq  x_m, \forall  x^e_k \in \mathcal{X}^e$ and $\|u^e_k\| \leq u_m, \forall  u^e_k\in\mathcal{U}^e_k$. If $\lambda$ and $\alpha$ are chosen such that
\begin{equation}
     x_m^2 + u_m^2 \leq \frac{2-\alpha}{\lambda};\quad \alpha \in (0,2); \quad \lambda >0
\end{equation}
then the constraint \eqref{eq:17} is satisfied $\forall ( x^e_k, u^e_k,k)\in \mathcal{X}^e \times \mathcal{U}^e_k\times \mathbb{I}^+$. 
\end{propos}
\subsection{State Error Induced by Parameter Adaptation}
The system \eqref{eq:incremental} can be rewritten as 
\begin{equation}
  x^e_{k+1}=\hat{A}_kx^e_k+\hat{B}_ku^e_k+\Tilde{x}^e_{k+1}
  \label{eq:es}
\end{equation}
\begin{lem} \label{lem:delx}
    Let $e \triangleq \Theta^d X^e$, where $\Theta^d \in \mathbb{R}^{n\times(n+m)}$ and $X^e=[x^{e^T} , u^{e^T}]^T \in \mathbb{R}^{n+m}$. When $\hat{\Theta}_k$ is updated following \eqref{eq:Adaptive} with $\hat{\Theta}_0\in \Psi$ and $\forall X^e=[x^{e^T} , u^{e^T}]^T,x^e\in\mathcal{X}^e$, and $u^e\in \mathcal{U}^e_k$, then the bound on $\Tilde{x}^e_k$ can be given by
\begin{equation}
    \|\tilde{x}^e_k\|\leq \sqrt{\delta_{\tilde{x}^e}}, \forall k \in \mathbb{I}^+
\end{equation}
where $\delta_{\Tilde{x}^e}$ can be computed from $\delta_{\Tilde{x}^e}=\text{max}\{e^Te|e\in \Gamma\}$, where $\Gamma\triangleq\{\Theta^d X^e|\Theta^d \in \mathcal{B}_{d_{\Theta}}(\textbf{0}),X^e=[x^{e^T} , u^{e^T}]^T,x^e\in \mathcal{X}^e,u^e\in\mathcal{U}^e_k\},d_\Theta=\text{diam}(\Psi)$, and $\textbf{0}\in \mathbb{R}^{n\times(n+m)}$ is the zero vector. 
\end{lem} 
\begin{lem} \label{lem:del}
    When $\hat{\Theta}_k$ is updated based on \eqref{eq:Adaptive} and for all instants the constraints \eqref{eq:IC} and \eqref{eq:RIC} are satisfied, then the following is true
\begin{equation}
    \|\Delta\hat{\Theta}_{k+1}X^e\|\leq\delta_\Theta=(2-\alpha)\sqrt{\delta_{\Tilde{x}^e}}
\end{equation}
$\forall X^e=[x^{e^T} , u^{e^T}]^T,x^e\in\mathcal{X}^e$, and $u^e\in \mathcal{U}^e_k$, where $\delta_{\Tilde{x}^e}$ and $\alpha$ are obtained from \textit{Lemma} \ref{lem:delx} and \textit{Proposition} \ref{propos:xm}, respectively.
\end{lem}  
\par\textit{Definition 2:} The $C_0-$set containing the errors (defined in \textit{lemmas} \ref{lem:delx} and \ref{lem:del}) are defined as $W_{\Tilde{x}^e}=\{\Tilde{x}^e_k|\|\Tilde{x}^e_k\|\leq \sqrt{\delta_{\Tilde{x}^e}},\Tilde{x}^e=x^e_k-\hat{x}^e_k,\forall x^e_k,\hat{x}^e_k\in \mathbb{R}^n,\forall k\in \mathbb{I}^+\}$ and $W_{\Theta}=\{\Delta \hat{\Theta}X^e |\|\Delta \hat{\Theta}X^e\|\leq \delta_\Theta, \Delta \hat{\Theta}=\hat{\Theta}_i-\hat{\Theta}_j,\forall \hat{\Theta}_{i,j}\in \Psi, \forall i,j\in \mathbb{I}^+,X^e \in \mathbb{R}^{(n+m)},X^{eT}X^e\leq \frac{2-\alpha}{\lambda}\}$.
\section{Controller Design}
This section provides insight into the controller design, including the formulation of the state feedback gain, the reformulation of input and input rate constraints in the error dynamics domain, the development of the terminal set, the formulation of the MPC optimization problem, and the analysis of recursive feasibility and closed-loop stability.
\subsection{State Feedback Gain}
\begin{assum}\label{asm:state feedback}
    Let $(P,K)\in \mathbb{R}^{n\times n} \times \mathbb{R}^{m \times n}$ be a pair associated with a positive definite matrix $Q \in \mathbb{R}^{n \times n}$ such that $\forall[\hat{A}_k,\hat{B}_k]\in \Psi$ and $\forall k \in \mathbb{I}^+$ the following holds
\begin{equation}
    P>0; \quad (\hat{A}_k+\hat{B}_kK)^TP(\hat{A}_k+\hat{B}_kK)-P+Q<0 
    \label{eq:P}
\end{equation}
 Moreover, there exists a set $\zeta \subset \mathcal{X}^e$ associated with a constant $\gamma \in (0,1)$ such that:
\end{assum}
\begin{equation}
    z_{k+1}^TPz_{k+1}<\gamma^2z_k^TPz_k, \text{where} 
    \label{eq:21}
\end{equation}
\begin{equation}
    z_{k+1}=(\hat{A}_k+\hat{B}_kK)z_k, \nonumber
\end{equation}
\begin{equation}
    \forall ([\hat{A}_k,\hat{B}_k],z_k,Kz_k,k)\in \Psi \times \zeta \times \mathcal{U}\times \mathbb{I}^+
\end{equation}
Condition \eqref{eq:21} assumes that $\zeta$ is a $\gamma$-contractive set with respect to the stabilizing feedback gain $K$, ensuring that $z_{k+1} \in \gamma \zeta$ for all $z_k \in \zeta$, where $\gamma \in (0,1)$. For a given $\gamma$, the $\gamma-$contractive set can be determined using the procedure outlined in \cite{Dhar2021}. The control input $u^e_k$ can be chosen as 
\begin{equation}
     u^e_k=K x^e_k+ v^e_k
    \label{eq:dual}
\end{equation}
\par\textit{Remark 1:} In the same lines, the control input rate can be chosen as 
\begin{equation}
    \Delta u^e_k=K\Delta x^e_k+ \Delta v^e_k
    \label{eq:24}
\end{equation}
The MPC algorithm developed in the subsequent section computes the control input $v^e_k$. Substituting \eqref{eq:dual} in \eqref{eq:incremental}, the following is obtained
\begin{equation}
    x^e_{k+1}= A^sx^e_k+Bv^e_k=\Theta^sX_k^v
    \label{eq:uncertain}
\end{equation}
where $A^s=A+BK$, $\Theta^s= [A^s,B]$, and $X_k^v=\begin{bmatrix}
    x^{eT}_k & v^{eT}_k
\end{bmatrix}^T$. In a similar manner, \eqref{eq:estimate} and \eqref{eq:compact estimate} can be modified as follows:
\begin{equation}
    \hat{x}^e_{k+1}=\hat{A}^s_kx^e_k+\hat{B}_k v^e_k=\hat{\Theta}^s_kX_k^v
\end{equation}
where $\hat{A}^s_k=\hat{A}_k+\hat{B}_kK$ and $\hat{\Theta}^s_k=\begin{bmatrix}
    \hat{A}^s_k & \hat{B}_k
\end{bmatrix}$. \\
\par \textit{Remark 2:} It is defined in \textit{Assumption} \ref{assum:lumped} that $\Psi$ is a $C-$set and $\Theta,\hat{\Theta}\in \Psi$ and K is a finite gain matrix, then there is a $C-$set $\Psi^a$ such that $A^s,\hat{A}^s_k\in \Psi^a, \forall k \in \mathbb{I}$, where $\Psi^a$ is 
\begin{equation}
    \Psi^a=co\{\psi^a_1,\psi^a_2,\cdots,\psi^a_L\};\psi^a_i\triangleq A_i+B_iK, \forall i \in \mathbb{I}_L^+
\end{equation}
where $\begin{bmatrix}
    A_i & B_i
\end{bmatrix},\forall i \in \mathbb{I}_L^+$ is defined in \textit{Assumption} \ref{assum:lumped}.
\subsection{Constraint Reformulation} 
\begin{lem}\label{lem:input rate}
    If \eqref{eq:se} is satisfied, then the state increment $\Delta x^e_k$ is bounded as
\begin{equation}
    \|\Delta x^e_k\| \le d_x:= \mathrm{diam}(\mathcal{X}^e), \quad \forall k \in \mathbb{I}^+
    \label{eq:xbound}
\end{equation}
\end{lem}
\par \textit{Definition 3:} The $C_0-$set containing the state rate (defined in \textit{Lemma} \ref{lem:input rate}) is $\mathcal{X}^e_\Delta =\{\Delta x^e_k |\|\Delta x^e_k\|\leq d_x, \Delta x^e_k=x^e_k-x^e_{k-1},\forall x^e_k\in \mathbb{R}^n,\forall k\in \mathbb{I}^+\}$ .
\\
\begin{lem}\label{lem:ue}
    At the $k^{th}$ instant if \eqref{eq:ue} and \eqref{eq:del} holds, then for all $i\in \mathbb{I}_{N-1}$ the following can be inferred
\begin{equation}
    u^e_{k+i\mid k}\hspace{-0.4em}\in\mathcal{U}^e_k:=\hspace{-0.4em}\left\{\, \hspace{-0.4em} \textbf{u}^e_k \in \mathbb{R}^{Nm}\hspace{-0.4em}\;\middle|\;
\begin{aligned}
\hspace{-0.4em}H_u \textbf{u}^e_k &\le h_u\pmb 1_{Nm}\\
\hspace{-0.4em}H_\Delta \textbf{u}^e_k &\le h^u_\Delta \pmb 1_{Nm} \hspace{-0.4em}+\bar{H} u^e_{k-1}
\end{aligned}\hspace{-0.4em}\right\} \label{eq:ueoverall}
\end{equation} 
where 
\begin{align}
    &H_u=\begin{bmatrix}
    I_m &0_m &\cdots &0_m\\
    0_m &I_m &\cdots &0_m\\
    \vdots &\vdots &\ddots &\vdots \\
    0_m & 0_m &\cdots &I_m
\end{bmatrix}\in \mathbb{R}^{Nm \times Nm} \nonumber \\
    & H_\Delta\hspace{-0.02in}=\hspace{-0.02in}\begin{bmatrix}
    I_m & 0_m &\cdots & 0_m & 0_m & 0_m\\
    -I_m & I_m &\cdots & 0_m & 0_m & 0_m\\
    \vdots &\vdots &\ddots &\vdots &\vdots &\vdots\\
    0_m & 0_m &\cdots &-I_m &I_m &0_m \\
    0_m & 0_m &\cdots &0_m &-I_m &I_m 
\end{bmatrix}\hspace{-0.02in}\in \hspace{-0.02in}\mathbb{R}^{Nm\times Nm} \nonumber \\
    & h_u\in \mathbb{R}; \quad h^u_\Delta\in \mathbb{R}  \nonumber \\
    & \bar{H}=\begin{bmatrix}
    I_m & \textbf{0}_{Nm-1\times m}
\end{bmatrix}^T\in \mathbb{R}^{Nm\times m}\nonumber
\end{align}
$\textbf{u}^e_k =\{u^e_{k|k},u^e_{k+1|k}\cdots,u^e_{k+N-1|k}\}^T$, $u^e_{k-1}$ is the control input of the preceding instant and $\pmb 1_{Nm}=\begin{bmatrix}
    1, 1, \cdots, 1
\end{bmatrix}^T\in \mathbb{R}^{Nm}$.
\end{lem}
\begin{proof} At the $k^{th}$ instant for all $ i\in \mathbb{I}_{N-1}$, the predicted control inputs $u^e_{k+i\mid k}$ must satisfy \eqref{eq:ue} for the system \eqref{eq:incremental}
\begin{equation}
    u^e_{k+i\mid k}\in \mathcal{U}^e:=\{\textbf{u}^e_k\;|\;H_u \textbf{u}^e_k \le h_u\pmb 1_{Nm}\}
    \label{eq:u}
\end{equation}
 At the $k^{th}$ instant for all $ i\in \mathbb{I}_{N-1}$, the predicted incremental control inputs $\Delta u^e_{k+i\mid k}$ must satisfy \eqref{eq:del} for the system \eqref{eq:incremental}
\begin{equation}
    \Delta u^e_{k+i\mid k} \in \mathcal{U}_\Delta : =\{\Delta \textbf{u}^e_k|H_u \Delta \textbf{u}^e_k \leq h^u_\Delta \pmb 1_{Nm}\}
    \label{eq:10}
\end{equation}
For the entire prediction horizon $\Delta u^e_{k+i\mid k}$ can be reformulated as
\begin{equation}
    \Delta u^e_{k+i\mid k}= u^e_{k+i\mid k}- u^e_{k+i-1\mid k}, \quad \forall i \in \mathbb{I}_{N-1}
    \label{eq:12}
\end{equation}
and for $i=0 \Rightarrow u^e_{k-1|k}=u^e_{k-1}$, which is the control input of the preceding instant and is measured at the $k^{th}$ instant. Utilizing \eqref{eq:10} and \eqref{eq:12} the following can be inferred
\begin{equation}
    u^e_{k+i\mid k}\in\mathcal{U}_\Delta:=\{\textbf{u}^e_k|H_\Delta \textbf{u}^e_k\leq h^u_\Delta \pmb 1_{Nm} + \bar{H}u^e_{k-1}\}
    \label{eq:13}
\end{equation}
The solution of the COCP $\mathbb{P}$, $\textbf{u}^e_k$ must satisfy \eqref{eq:u} and \eqref{eq:13} for all $i\in\mathbb{I}_{N-1}$
\begin{equation}
    u^e_{k+i\mid k}\hspace{-0.4em}\in \mathcal{U}^e_k:=\hspace{-0.4em}\left\{\,\hspace{-0.4em}\textbf{u}^e_k \in \mathbb{R}^{Nm} \;\hspace{-0.4em}\middle|\;
\begin{aligned}
\hspace{-0.4em}H_u \textbf{u}^e_k &\le h_u\pmb 1_{Nm}\\
\hspace{-0.4em}H_\Delta \textbf{u}^e_k &\le h^u_\Delta\pmb 1_{Nm} + \bar{H} u^e_{k-1}
\end{aligned}\hspace{-0.4em}\right\}\nonumber 
\end{equation}
This concludes the proof.
\end{proof}
\par \textit{Remark 3:} The admissible set $\mathcal{U}^e_k$ is time varying due to its dependency on previous input $u^e_{k-1}$, but it still remains a convex polytope. 
\begin{lem}\label{lem:ve}
    At the $k^{th}$ instant if \eqref{eq:se}, \eqref{eq:ue} and \eqref{eq:del} holds, then the following can be inferred for all $i\in\mathbb{I}_{N-1}$ 
\begin{equation}
    \hspace{-0.4em}v^e_{k+i\mid k}\hspace{-0.4em}\in \mathcal{V}^e_k:=\hspace{-0.4em}\left\{\,\hspace{-0.4em}\textbf{v}^e_k \in \mathbb{R}^{Nm}\hspace{-0.4em}\;\middle|\;
\begin{aligned}
\hspace{-0.4em}H_u \textbf{v}^e_k &\le h_v\pmb 1_{Nm}\\
\hspace{-0.4em}H_\Delta \textbf{v}^e_k &\le h^v_\Delta\pmb 1_{Nm} + \bar{H} v^e_{k-1}
\end{aligned}\hspace{-0.4em}\right\}
\end{equation}
where $h_v,h^v_\Delta\in\mathbb{R}$, $v^e_{k-1}$ is the optimized control input of the preceding instant and $\textbf{v}^e_k=\{v^e_{k\mid k},v^e_{k+1\mid k},\cdots,v^e_{k+N-1\mid k}\}^T$.
\end{lem}
\begin{proof} Utilizing \eqref{eq:dual}, if \eqref{eq:se} and \eqref{eq:ue} holds then at $k^{th}$ instant for all $i\in \mathbb{I}_{N-1}$, the following can be inferred
\begin{equation}
    v^e_{k+i\mid k} \in \mathcal{V}^e\triangleq\mathcal{U}^e\ominus k\mathcal{X}^e:=\{\textbf{v}^e_k|H_u \textbf{v}^e_k\leq h_v\pmb 1_{Nm}\} 
    \label{eq:34}
\end{equation}
In the same lines, at $k^{th}$ instant for all $i\in \mathbb{I}_{N-1}$ utilizing \eqref{eq:24} and \textit{Definition 3}, the following can be inferred for $\Delta v^e_{k+i\mid k}$, if \eqref{eq:del} holds
\begin{equation}
    \Delta v^e_{k+i\mid k}\in\mathcal{V}^e_\Delta \triangleq\mathcal{U}^e_\Delta \ominus K\mathcal{X}^e_\Delta:=\{\Delta \textbf{v}^e_k| H_u\Delta \textbf{v}^e_k\leq h^v_\Delta\pmb 1_{Nm}\}
    \label{eq:delv}
\end{equation}
where $\Delta \textbf{v}^e_k=\{\Delta v^e_{k\mid k},\Delta v^e_{k+1\mid k},\cdots,\Delta v^e_{k+N-1\mid k}\}^T$. For the entire prediction horizon, $\Delta v^e_{k+i\mid k}$ can be reformulated as 
\begin{equation}
    \Delta v^e_{k+i\mid k}= v^e_{k+i\mid k}- v^e_{k+i-1\mid k}, \quad \forall i \in \mathbb{I}_N
    \label{eq:36}
\end{equation}
$v^e_{k-1\mid k}=v^e_{k-1}$ is the optimal control input of the preceding instant and is a measured value at the $k^{th}$ instant. From \eqref{eq:delv} and \eqref{eq:36} the following can be obtained
\begin{equation}
    v^e_{k+i\mid k}\in\mathcal{V}^e_\Delta:=\{\textbf{v}^e_k|H_\Delta \textbf{v}^e_k\leq h^v_\Delta\pmb 1_{Nm} + \bar{H}v^e_{k-1}\}
    \label{eq:37}
\end{equation}
Therefore, a combined constraint at $k^{th}$ instant for all $i\in \mathbb{I}_{N-1}$ on $v^e_{k+i\mid k}$ can be given by
\begin{equation}
    v^e_{k+i\mid k}\hspace{-0.4em}\in \mathcal{V}^e_k:=\hspace{-0.4em}\left\{\hspace{-0.4em}\, \textbf{v}^e_k \in \mathbb{R}^{Nm}\hspace{-0.4em}\;\middle|\;
\begin{aligned}
\hspace{-0.4em}H_u \textbf{v}^e_k &\le h_v\pmb 1_{Nm}\\
\hspace{-0.4em}H_\Delta \textbf{v}^e_k &\le h^v_\Delta\pmb 1_{Nm} + \bar{H} v^e_{k-1}
\end{aligned}\hspace{-0.4em}\right\}
\end{equation}
This concludes the proof.
\end{proof}
\par For the uncertain system \eqref{eq:uncertain}, the objective function $J_k$ is defined by $J_k=\sum_{i=0}^{N-1}\frac{1}{2}[\|x^e_{k+i \mid k}\|_Q^2+\|v^e_{k+i \mid k}\|_R^2]+J_f( x^e_{k+N \mid k})$, where Q and R are positive definite matrices. The terminal constraint $J_f$ is defined as $J_f\triangleq\frac{1}{2}\| x^e_{k+N \mid k}\|_P^2$, where P is chosen from \eqref{eq:P}. The COCP $\mathbb{P}_1$ for the uncertain system \eqref{eq:uncertain} is defined as 
\begin{equation}
    \mathbb{P}_1: \textbf{v}^{e^*}_k=\text{argmin}_{v^e_k}J_k \nonumber
\end{equation}
\begin{subequations}
\begin{align}
    \text{s.t.} \quad &x^e_{k+i \mid k} 
    = A^s x^e_{k+i-1 \mid k} + Bv^e_{k+i-1\mid k}, 
     \forall i \in \mathbb{I}_{N}^+ 
    \label{eq:27a}\\
    & x^e_{k+i \mid k} \in \mathcal{X}^e,   
     v^e_{k+i-1\mid k} \in \mathcal{V}^e_k, \forall i \in \mathbb{I}_N^+
     \\
    &x^e_{k+N \mid k} \in \mathcal{X}_T \subset \mathcal{X}^e  
\end{align}
\end{subequations}
 At $k^{th}$ instant the optimization routine solution is $ \textbf{v}_k^{e^*}=\{v^{e^{*}}_{k \mid k},v^{e^{*}}_{k+1 \mid k},\cdots, v^{e^{*}}_{k+N-1 \mid k}\}^T$ and the terminal set is $\mathcal{X}_T$. Since equation \eqref{eq:27a} depends on uncertain system parameters, it is not implementable. An estimated system model is introduced via the following proposition, which can be used in $\mathbb{P}_1$ as an alternative to \eqref{eq:27a}.
 \begin{propos}
    Based on \eqref{eq:es}, the prediction model \eqref{eq:27a} is equivalent to the following model
\begin{equation}
    x^e_{k+i\mid k}\hspace{-0.4em}=\hspace{-0.4em}\hat{A}^s_{k+i-1}x^e_{k+i-1\mid k}+\hspace{-0.4em}\hat{B}_{k+i-1}v^e_{k+i-1\mid k}+\Tilde{x}^e_{k+i \mid k}, \forall i \in \mathbb{I}_N^+
    \label{eq:pred}
\end{equation}where, $\hat{A}^s_{k+(.)}\hspace{-0.4em}=\hspace{-0.4em}\hat{A}_{k+(.)}+\hat{B}_{k+(.)}K$, $\Tilde{x}^e_{k+i\mid k}\hspace{-0.4em}=\hspace{-0.4em}\Tilde{\Theta}^s_{k+i-1}X^v_{k+i-1\mid k}$, $\Tilde{\Theta}^s_{k+(.)}\hspace{-0.4em}\triangleq\Theta^s-\hat{\Theta}^s_{k+(.)}$, and $X^v_{k+(.)\mid k}=\begin{bmatrix}
    x^{e^T}_{k+(.)\mid k} & v^{e^T}_{k+(.)\mid k}
\end{bmatrix}$.
 \end{propos}
\par \textit{Remark 4:} At the $k^{th}$ instant for all $i\in \mathbb{I}_N^+$ the knowledge of $\hat{\Theta}^s_{k+i-1}$ and $\Tilde{x}^e_{k+i}$ are not known. Hence \eqref{eq:pred} is unimplementable. 
\begin{propos}
    If \eqref{eq:Adaptive} updates $\hat{\Theta}_k$ with $X^e_{k+i\mid k}$ satisfying \eqref{eq:se} and \eqref{eq:ueoverall}, then $x^e_{k+i\mid k}$ can be segregated as 
\begin{equation}
    x^e_{k+i\mid k}=\bar{s}^e_{k+i\mid k}+\epsilon_{k+i\mid k}+e_{k+i\mid k}, \forall i\in \mathbb{I}_N^+
    \label{eq:sepa}
\end{equation}
where $\epsilon_{k+i\mid k}\in \mathbb{Z}_i^{\Theta}$ and $e_{k+i\mid k}\in \mathbb{Z}^{\Tilde{x}^e}$ capture the errors induced by parameter estimation and disturbance term $\Tilde{x}^e_{k+i\mid k}$ respectively.
\end{propos}
\par \textit{Remark 5:} The sets $\mathbb{Z}_i^{\Theta}$ and $\mathbb{Z}^{\tilde{x}^e}$ are obtained from the sets $W_{\tilde{x}^e}$ and $W_\Theta$ defined in \textit{Definition 2}. The computation follows the methodology presented in [\textit{Section V-B}, \cite{Dhar2021}].
\subsection{Terminal Constraints}
The system given by $ x^e_{k+1}=\hat{A}^s_k x^e_k,\hat{A}^s_k\in \Psi^a, \forall k\in \mathbb{I}^+$ can be claimed to be exponentially stable by utilizing \textit{Assumption }\ref{asm:state feedback}. There is a need to obtain $\gamma-$contractive set $\mathcal{X}^e_T$ as follows
\begin{equation}
    \mathcal{X}^e_T\subset \mathcal{X}^e\ominus \mathbb{Z}^{\Tilde{x}^e}
\end{equation}
\begin{equation}
    \hat{A}^s_k s^e_k\in \gamma\mathcal{X}^e_T, \forall  s^e_k\in \mathcal{X}^e_T, \forall \hat{A}^s_k \in \Psi^a, \forall k \in \mathbb{I}
\end{equation}
The value of $\gamma$ is chosen between $(0,1)$. If $K(\mathcal{X}^e_T\oplus\mathbb{Z}^{\Tilde{x}^e})\subset \mathcal{U}^e_k$, only then is $\mathcal{X}^e_T$ said to be feasible.
\par The optimized control input sequence $\textbf{v}^e_k$ is subjected to the following terminal constraint
\begin{equation}
    v^e_{k+N-1\mid k}\geq - h^v_\Delta \pmb 1_m
    \label{eq:terminal}
\end{equation}
where $\pmb 1_m=\begin{bmatrix}
    1, 1, \cdots, 1
\end{bmatrix}^T\in \mathbb{R}^m$. 
\subsection{Reformulated MPC Problem}
The constraints on $\bar{s}^e_{k+i \mid k}$ and $\bar{s}^e_{k+N \mid k}$ after tightening are given by
\begin{equation}
    \begin{aligned}
        \bar{s}^e_{k+i \mid k} &\in\mathcal{X}^e_i\triangleq \mathcal{X}^e\ominus \mathbb{Z}_i^{\Theta}\ominus \mathbb{Z}^{\tilde{x}^e},  \quad \forall i \in \mathbb{I}_{N-1}^+\\
        \bar{s}^e_{k+N \mid k} &\in \bar{\mathcal{X}}^e_T \triangleq \mathcal{X}^e_T \ominus \mathbb{Z}_N^{\Theta}\subset \mathcal{X}^e\ominus \mathbb{Z}_N^{\Theta}\ominus \mathbb{Z}^{\tilde{x}^e} 
    \end{aligned}
    \label{eq:43}
\end{equation}
The MPC optimization routine designed for the uncertain system is defined as follows 
\begin{equation}
   \begin{aligned}
        \bar{s}^e_{k+i \mid k}&=\hat{A}^s_k\bar{s}^e_{k+i-1 \mid k}
         +\hat{B}_kv^e_{k+i-1 \mid k}, \forall i \in \mathbb{I}_N^+
    \label{eq:nominal}
   \end{aligned}
\end{equation}
The cost function for the system \eqref{eq:nominal} is defined as $J^s_k=\sum_{i=0}^{N-1}\frac{1}{2}[\|\bar{s}^e_{k+i \mid k}\|_Q^2+\| v^e_{k+i \mid k}\|_R^2]+J^s_f(\bar{s}^e_{k+N \mid k})$, where $J^s_f(\bar{s}^e_{k+N \mid k})\triangleq\frac{1}{2}\|\bar{s}^e_{k+N \mid k}\|_P^2$. The modified COCP $\mathbb{P}_1$ is as follows 
\begin{equation} 
    \mathbb{P}_2: v^{e^*}_k=\text{argmin}_{v^e_k}J_k^s \nonumber
\end{equation}
\begin{subequations}
\begin{align}
    \text{s.t. } &\bar{s}^e_{k+i \mid k} 
    = \hat{A}^s_k  \bar{s}^e_{k+i-1 \mid k} + \hat{B}_k v^e_{k+i-1 \mid k}, 
     \forall i \in \mathbb{I}_{N}^+ 
    \label{eq:36a}\\
    &\bar{s}^e_{k+i \mid k} \in \mathcal{X}^e_i,
     v^e_{k+i-1\mid k} \in \mathcal{V}^e_k, \quad \forall i\in \mathbb{I}_{N}^+
     \\
    &v^e_{k+N-1\mid k}\geq -h^v_\Delta\textbf{1}_m,  \bar{s}^e_{k+N \mid k} \in \bar{\mathcal{X}}^e_T  
\end{align}
\end{subequations}
\subsection{Recursive Feasibility}
\begin{lem}\label{lem:feasible}
    If $\mathbb{P}_2$ is initially feasible at $k^{th}$ instant, then there exists a control input sequence $\textbf{v}^e_{k+1}$ at $k+1^{th}$ instant such that
\begin{equation}
    \hspace{-0.4em}\textbf{v}^e_{k+1}\hspace{-0.4em}  \in \mathcal{V}^e_{k+1}\hspace{-0.4em} :=\left\{\,\hspace{-0.4em}  \textbf{v}^e_{k+1} \in \mathbb{R}^{Nm}\hspace{-0.4em}  \;\middle|\;
\begin{aligned}
\hspace{-0.4em} H_u \textbf{v}^e_{k+1} &\le h_v\pmb 1_{Nm}\\
\hspace{-0.4em} H_\Delta \textbf{v}^e_{k+1} &\le h^v_\Delta\pmb 1_{Nm} + \bar{H} v^e_{k}
\end{aligned}\hspace{-0.4em}
\right\}
\label{eq:recur}
\end{equation}
where $\textbf{v}^e_{k+1}\triangleq\{v^{e}_{k+1\mid k},\cdots,v^{e}_{k+N-1\mid k},0_m\}^T$, and $v^e_k$ is the optimal control input of the $k^{th}$ instant. 
\end{lem}

\begin{proof}
The initial feasibility of $\mathbb{P}_2$ at $k^{th}$ instant guarantees $\textbf{v}^e_k\in \mathcal{V}^e_k$. Therefore, at $k+1^{th}$ instant, with $v^e_{k+i|k+1}=v^e_{k+i|k}, \forall i\in \mathbb{I}^+_{N-1}$ and utilizing \eqref{eq:delv} and \eqref{eq:36} the following can be inferred 
\begin{equation}
    H_u(v^e_{k+i|k+1}-v^e_{k+i-1|k+1}) \leq h^v_\Delta\pmb 1_{Nm} ,\quad  \forall i \in \mathbb{I}_{N-1}^+
    \label{eq:53a}
\end{equation}
At $k+1^{th}$ instant, $v^e_k$ is the implemented control input from the preceding instant. Utilizing \eqref{eq:37}, \eqref{eq:53a} can be rewritten as
\begin{equation}
    H_\Delta v^e_{k+i|k+1}\leq h^v_\Delta\pmb 1_{Nm} + \bar{H}v^e_k, \quad \forall i \in \mathbb{I}_{N-1}^+
    \label{eq:54}
\end{equation}
Since $v^e_{k+N|k+1}=0_m$, the following is concluded by invoking \eqref{eq:terminal}, 
\begin{equation}
    v^e_{k+N\mid k+1}-v^e_{k+N-1\mid k+1}\leq h^v_\Delta\textbf{1}_m
\end{equation}
Therefore utilizing \eqref{eq:34}, \eqref{eq:delv}, \eqref{eq:terminal} and \eqref{eq:54}, it can be concluded that \eqref{eq:recur} holds.
\end{proof}
\par \textit{Remark 6:} The term $\bar{H}v^e_k$ in equation \eqref{eq:54} is responsible for making the set $\mathcal{V}^e_{k+1}$ in equation \eqref{eq:recur}(and hence the input constraint) time varying in nature. 
\begin{lem} \label{lem:Recursive Feasibility}
    At the $k^{th}$ instant, if $x^e_k\in \mathcal{X}^e$ and $\bar{s}^e_{k\mid k}\in\mathcal{X}^e_0$, then the feasibility of the COCP $\mathbb{P}_2$ ensures the following
\begin{enumerate}
    \item $x^e_{k+i\mid k}\in \mathcal{X}^e, \forall i \in \mathbb{I}_N^+$ 
    \item $x_{k+i\mid k}\in \mathcal{X}, \forall i \in \mathbb{I}_N^+$ 
\end{enumerate} 
\end{lem}
\begin{proof} At the $k^{th}$ instant, $\textbf{v}^e_k$ is considered to be the solution of $\mathbb{P}_2$. The following has been assumed for a value of i, where $i\in \mathbb{I}_N^+$
\begin{equation}
    x^e_{k+j-1\mid k}\in \mathcal{X}^e, \Delta x^e_{k+j-1\mid k}\in \mathcal{X}^e_\Delta,\forall j \in \mathbb{I}_i^+
    \label{eq:46}
\end{equation}
$\textbf{v}^e_k\in \mathcal{V}^e_k$ is guaranteed as $\mathbb{P}_2$ is initially feasible, which along with \eqref{eq:dual}, \eqref{eq:24} and \eqref{eq:46} implies $u^e_{k+j-1\mid k}\in \mathcal{U}^e_k$. Since \eqref{eq:Adaptive} updates the estimated parameters, and $X^e=[x^{e^T},  u^{e^T}]^T$ where $x^e\in \mathcal{X}^e$ and $u^e\in \mathcal{U}^e_k$ then the following holds for all $j\in \mathbb{I}_i, i\in \mathbb{I}_N^+$
\begin{equation}
    e_{k+j}\in \mathbb{Z}^{\Tilde{x}^e}, \epsilon_{k+j \mid k}\in \mathbb{Z}^\Theta_j 
\end{equation}
For the value of $i\in \mathbb{I}_N^+$ utilized in \eqref{eq:46} the following can be concluded
\begin{equation}
    \bar{s}^e_{k+j\mid k}+\epsilon_{k+j\mid k}+e_{k+j\mid k}\in \bar{s}^e_{k+j\mid k}\oplus\mathbb{Z}^\Theta_j\oplus\mathbb{Z}^{\Tilde{x}^e}
    \label{eq:48}
\end{equation}
At $k^{th}$ instant if $\mathbb{P}_2$ is feasible and \eqref{eq:46} holds, then using \eqref{eq:sepa}, \eqref{eq:43} and \eqref{eq:48} the following can be inferred
\begin{equation}
    x^e_{k+i\mid k}\in \mathcal{X}^e, \text{if } i\in \mathbb{I}_N^+
    \label{eq:49}
\end{equation}
Therefore, if \textit{Assumption 2} holds and $x^e_k\triangleq x_k-x^r_k$, then utilizing \eqref{eq:49} the following is concluded
\begin{equation}
    x_{k+i\mid k}\in \mathcal{X}, \text{if } i\in \mathbb{I}_N^+
    \label{eq:50}
\end{equation}
Hence, if $x^e_{k\mid k}\in \mathcal{X}^e$ and $\bar{s}^e_{k\mid k}\in \mathcal{X}^e_0$ then at $k^{th}$ instant by recursively using \eqref{eq:46}-\eqref{eq:50} the claimed assertions hold. 
\end{proof}
\begin{thm}\label{thm:recur}
    The recursive feasibility of $\mathbb{P}_2$ for $(k+i)^{th}$ instants for all $i\in \mathbb{I}^+$ is guaranteed if the COCP $\mathbb{P}_2$ is feasible at the $k^{th}$ instant.
\end{thm} 
\begin{proof} At $k^{th}$ instant it is assumed that $\mathbb{P}_2$ is feasible and provides a feasible solution $\textbf{v}^e_k$. Hence, at $k^{th}$ instant the following is guaranteed $\bar{s}^e_{k+i\mid k}\in \mathcal{X}^e_i, \forall i \in \mathbb{I}_{N-1}^+$ and $\bar{s}^e_{k+N\mid k}\in\bar{\mathcal{X}}^e_T$, where $\mathcal{X}^e_i$ and $\bar{\mathcal{X}}^e_T$ is defined by \eqref{eq:43}. From \textit{Lemma }\ref{lem:feasible} it is inferred that there exists a control input sequence $\textbf{v}^e_{k+1}\in\mathcal{V}^e_{k+1}$. Hence, the recursive feasibility of $\mathbb{P}_2$ is guaranteed if  \begin{equation}
    \bar{s}^e_{k+1+i\mid k+1}\in \mathcal{X}^e_i, \forall i\in \mathbb{I}_{N-1}; \quad \bar{s}^e_{k+1+N\mid k+1}\in \bar{\mathcal{X}}_T
\end{equation}
The remaining proof proceeds along the same lines as provided in [\textit{Section V-E},\cite{Dhar2021}].
\end{proof}
\par \textit{Corollary 1:} From \textit{lemma }\ref{lem:Recursive Feasibility} it is proved that at $k^{th}$ instant the feasibility of $\mathbb{P}_2$ guarantees $x_{k+i \mid k}\in \mathcal{X}$. Hence, the recursive feasibility of $\mathbb{P}_2$ is guaranteed if there exists a feasible solution that $x_{k+j+i\mid k}\in \mathcal{X}, \forall(i,j)\in \mathbb{I}_N\times\mathbb{I}^+$. 
\subsection{Stability Analysis}
\textit{Remark 7:} Let $\bar{\alpha_1}(\cdot)$ and $\bar{\alpha_2}(\cdot)$ be class $\mathcal{K}$ functions such that:
\begin{equation}
    \bar{\alpha_1}(\|\bar{s}^e_k\|_Q^2+\|v^e_k\|_R^2)\leq J^s_k \leq \bar{\alpha_2}(\|\bar{s}^e_k\|_Q^2+\|v^e_k\|_R^2)
\end{equation}
\begin{thm}\label{thm:stability}
    Provided $\mathbb{P}_2$ is feasible, the closed loop system \eqref{eq:nominal} is bounded and asymptotically converging to zero. 
\end{thm}
\begin{proof} A Lyapunov function candidate for the system \eqref{eq:nominal} is constructed as 
\begin{equation}
    V_k=\hspace{-0.4em}J^s_k\hspace{-0.4em}=\hspace{-0.4em}\sum_{i=0}^{N-1}\frac{1}{2}[\|\bar{s}^e_{k+i\mid k}\|_Q^2+\hspace{-0.4em}\|v^e_{k+i\mid k}\|_R^2]+\hspace{-0.4em}J^s_f(\bar{s}^e_{k+N\mid k})
    \label{eq:57}
\end{equation}
At $k^{th}$ instant the optimal control input sequence $\textbf{v}^e_{k}$ is obtained by solving the cost function $J^s_k$, having $\bar{s}^e_{k+i\mid k}$ as the optimal states for all $i\in \mathbb{I}_N^+$. A feasible solution $\textbf{v}^e_{k+1}$ is available at the $k+1^{th}$ instant, since at the $k^{th}$ instant $\mathbb{P}_2$ is feasible (using \textit{Lemmas \ref{lem:feasible},\ref{lem:Recursive Feasibility}} and \textit{Theorem \ref{thm:recur}}). The cost function $J^s_{k+1}$ and a scalar function $V_{k+1}$ corresponding to the solution $\textbf{v}^e_{k+1}$ is defined as
\begin{equation}
    \begin{aligned}
        &V_{k+1}=J^s_{k+1}=\sum_{i=1}^{N-1}\frac{1}{2}[\|\bar{s}^e_{k+i\mid k+1}\|_Q^2+\|v^e_{k+i\mid k}\|_R^2]+ \\&\frac{1}{2}[\|\bar{s}^e_{k+N\mid k+1}\|_Q^2+\|v^e_{k+N\mid k+1}\|_R^2]+J^s_f(\bar{s}^e_{k+1+N\mid k+1})
    \end{aligned}
\end{equation}
On implementing $\textbf{v}^e_{k+1}$ in \eqref{eq:nominal}, $\bar{s}^e_{k+i\mid k+1}$ are obtained. Let $\Delta V_{k+1}=V_{k+1}-V_k$. This yields:
\begin{equation}
    \begin{aligned}
        &\Delta V_{k+1}\hspace{-0.4em}=-\frac{1}{2}[\|\bar{s}^e_{k\mid k}\|_Q^2+\hspace{-0.4em}\|v^e_{k\mid k}\|_R^2]\hspace{-0.4em}+\hspace{-0.4em}\frac{1}{2}\|\bar{s}^e_{k+N\mid k}\|_Q^2-\hspace{-0.4em}J^s_f(\bar{s}^e_{k+N\mid k})\\&+\hspace{-0.4em}J^s_f(\bar{s}^e_{k+1+N\mid k+1})\hspace{-0.4em}+\hspace{-0.4em}\sum_{i=1}^N\frac{1}{2}[\|\epsilon^1_{k+i\mid k}\|_Q^2+\hspace{-0.4em}2\bar{s}^{e^T}_{k+i\mid k}Q\epsilon^1_{k+i\mid k}]\hspace{-0.4em}
    \end{aligned}
    \label{eq:question}
\end{equation}
$\epsilon^1_{k+i\mid k}$ is the error due to the difference in the prediction model $\hat{\Theta}_k$ at $k^{th}$ instant and $\hat{\Theta}_{k+1}$ at $k+1^{th}$ instant with respect to the input sequence $\textbf{v}^e_{k+1}$, defined as 
$\epsilon^1_{k+i\mid k}=\bar{s}^e_{k+i\mid k+1}-\bar{s}^e_{k+i \mid k}$. For the feasible solution $\textbf{v}^e_{k+1}$, $\bar{s}^e_{k+1+N\mid k+1}=\hat{A}^s_{k+1}\bar{s}^e_{k+N\mid k+1}=\hat{A}^s_{k+1}(\bar{s}^e_{k+N\mid k}+\epsilon^1_{k+N\mid k})$. From \eqref{eq:P} the following can be inferred $(\bar{s}^e_{k+N\mid k}(\hat{A}^s_{k+1}P\hat{A}^{s^T}_{k+1}-P+Q)\bar{s}^e_{k+N\mid k})<0$, that implies 
\begin{equation}
    \begin{aligned}
        &\frac{1}{2}\|\bar{s}^e_{k+N\mid k}\|_Q^2+J^s_f(\bar{s}^e_{k+1+N\mid k+1})-J^s_f(\bar{s}^e_{k+N \mid k})\leq \\&\frac{1}{2}\|\hat{A}^s_{k+1}\epsilon^1_{k+N\mid k}\|_P^2+\bar{s}^{e^T}_{k+N\mid k}\hat{A}^{s^T}_{k+1}P\hat{A}^s_{k+1}\epsilon^1_{k+N\mid k} \nonumber
    \end{aligned}
\end{equation}
The above inequality is utilized to modify \eqref{eq:question} as follows
\begin{align}\label{eq:stabili}
\Delta V_{k+1}&=-\frac{1}{2}[\|\bar{s}^e_{k\mid k}\|_Q^2+\|v^e_{k\mid k}\|_R^2]+\frac{1}{2}\|\hat{A}^s_{k+1}\epsilon^1_{k+N\mid k}\|_P^2 \nonumber \\
&+\bar{s}^{e^T}_{k+N\mid k}\hat{A}^{s^T}_{k+1}P\hat{A}^s_{k+1}\epsilon^1_{k+N\mid k}\sum_{i=1}^N\frac{1}{2}[\|\epsilon^1_{k+i\mid k}\|_Q^2\nonumber \\
&+2\bar{s}^{e^T}_{k+i\mid k}Q\epsilon^1_{k+i\mid k}]
\end{align}
As $\mathbb{P}_2$ is recursively feasible, $\bar{s}^e_{k+i\mid k}\in\mathcal{X}^e_i\subset\mathcal{X}^e, \forall i\in\mathbb{I}_{N-1}^+$ and $\bar{s}^e_{k+N\mid k}\in \bar{\mathcal{X}}^e_T\subset \mathcal{X}^e$, hence $\|\bar{s}^e_{k+i\mid k}\|\leq x_m, \forall i \in \mathbb{I}_N^+$, where $x_m$ is defined in \textit{Proposition 2}. Moreover, $\hat{A}^s_k\in \Psi^a, \forall k\in \mathbb{I}^+$, where $\Psi^a$ is a $C$-set, then $\|\hat{A}^s_k\|_F\leq a$, where $a\in \mathbb{R}$ is a scalar. The class $\mathcal{K}$ function $\rho(\|L_k\|)$ can be defined as
\begin{equation}
    \begin{aligned}
        &\rho(\|L_k\|)\triangleq \frac{1}{2}(a^2\|P\|_F\|\epsilon^1_{k+N\mid k}\|^2+\sum_{i=1}^N[\|Q\|_F\|\epsilon^1_{k+i\mid k}\|^2 \\ &+2x_m\|Q\|_F\|\epsilon^1_{k+i\mid k}\|])+a^2x_m\|P\|_F\|\epsilon^1_{k+N\mid k}\|
    \end{aligned}
    \label{eq:stab}
\end{equation}
The upper bound of $\|L_k\|$ is used to obtain the function $\rho(\|L_k\|)$. From $\textit{Remark 7}$ and \eqref{eq:57}, the following can be inferred
\begin{equation}
    \|\bar{s}^e_{k\mid k}\|_Q^2+\|v^e_{k\mid k}\|_R^2 \geq \bar{\alpha}^{-1}_2(J^s_k)=\alpha_2^{-1}(V_k)
    \label{eq:stability}
\end{equation}
$\bar{\alpha}_2^{-1}(.)$ is the inverse function of $\alpha_2$, and therefore is a class $\mathcal{K}$ function. Utilizing \eqref{eq:stab} and \eqref{eq:stability}, \eqref{eq:stabili} is further modified into 
\begin{equation}
    \Delta V_{k+1}\leq \frac{1}{2}\bar{\alpha}_2^{-1}(V_k)+\rho(\|L_k\|)
    \label{eq:lya}
\end{equation} 
From \cite{Dhar2021} it is seen that $\epsilon^1_{k+i \mid k}$ is bounded and $\epsilon^1_{k+i\mid k}\rightarrow 0_n, \forall i\in \mathbb{I}_N^+$ as $k\rightarrow 0$. Hence, from \eqref{eq:stab} it can be inferred that $\rho(\|L_k\|)$ is bounded and $\rho(\|L_k\|)\rightarrow 0$ as $k\rightarrow 0$. 
\end{proof}
\par \textit{Corollary 2:} Using \eqref{eq:sepa}, the following is claimed
\begin{equation}
    x^e_k= \bar{s}^e_k+\epsilon_k+e_k
\end{equation}
If \eqref{eq:Adaptive} updates $\hat{\Theta}_k$ and $\mathbb{P}_2$ is feasible then using \cite{Dhar2021} and \textit{Theorem }\ref{thm:stability} it is guaranteed that $\epsilon_k,e_k,\bar{s}^e_k$ are bounded and $\epsilon_k,e_k,\bar{s}^e_k\rightarrow 0_n$, as $k\rightarrow \infty$. This implies $x^e_k$ is bounded for all $k\in\mathbb{I}$ and $x^e_k\rightarrow 0_n$ as $k\rightarrow\infty$. Therefore, provided \textit{Assumption 2} holds and $x^e_k=x_k-x^r_k$, it is claimed that $x_k$ is ultimately bounded.  

\section{Simulation Results}
The developed control framework is implemented on a stable second-order discrete-time plant described by $A=[-0.5,0;0.5,-0.4]$ and $B=[1,0;0,1]$. The uncertain system is expected to track a time-varying reference signal. The model parameters of the uncertain system belongs to the set $\Psi$, given as $\Psi=co\{[A_i,B_i]\},i\in\mathbb{I}_4^+$, where $A_1=[-0.6, 0; 0.35, -0.5]$, $ A_2=[-0.4, 0; 0.35, -0.32]$, $A_3=[-0.6 ,0.15; 0.6, -0.5]$ and $A_4=[-0.4, 0.15; 0.6, -0.32]$ and $B_1=[0.8, 0; 0, 0.8]$, $B_2=[1.2, 0; 0,1.2]$, $B_3=[0.8,0.1; 0.1,0.8]$ and $B_4=[1.2, 0.1; 0.1, 1.2]$. The imposed state, input and input rate constraints are $\|x\|_\infty\leq2.25$, $\| u_k\|_\infty \leq 3$ and $\|\Delta u_k\|_\infty \leq 2.5$ respectively, where $x_k=[x^{1}_k;x^{2}_k]^T$. The initial estimated model parameters are $\hat{A}_0=[-0.4,0.15;0.35,-0.5]$ and $\hat{B}_0=[1 ,0.1;0.1,0.8]$. The prediction horizon length $N$ is set to 5. To assess the performance of the proposed control framework under various conditions, simulations are performed with 10 different initial parameter estimates and 10 different initial states distributed throughout the prescribed constraint region.  
\begin{figure}
    \centering
    \includegraphics[width=\linewidth]{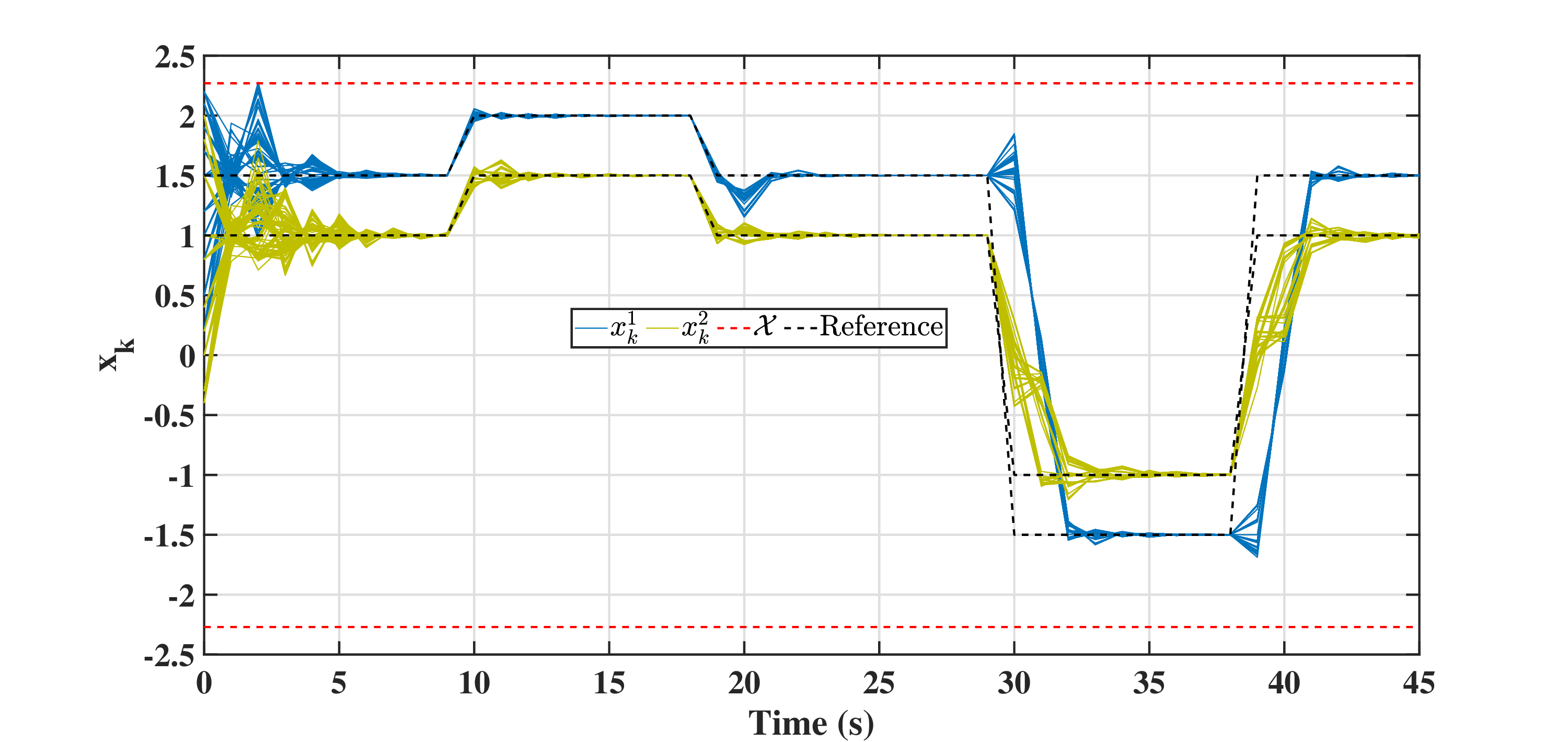}
    \caption{System States}
    \label{fig:1}
\end{figure}

\begin{figure}
    \centering
    \includegraphics[width=\linewidth]{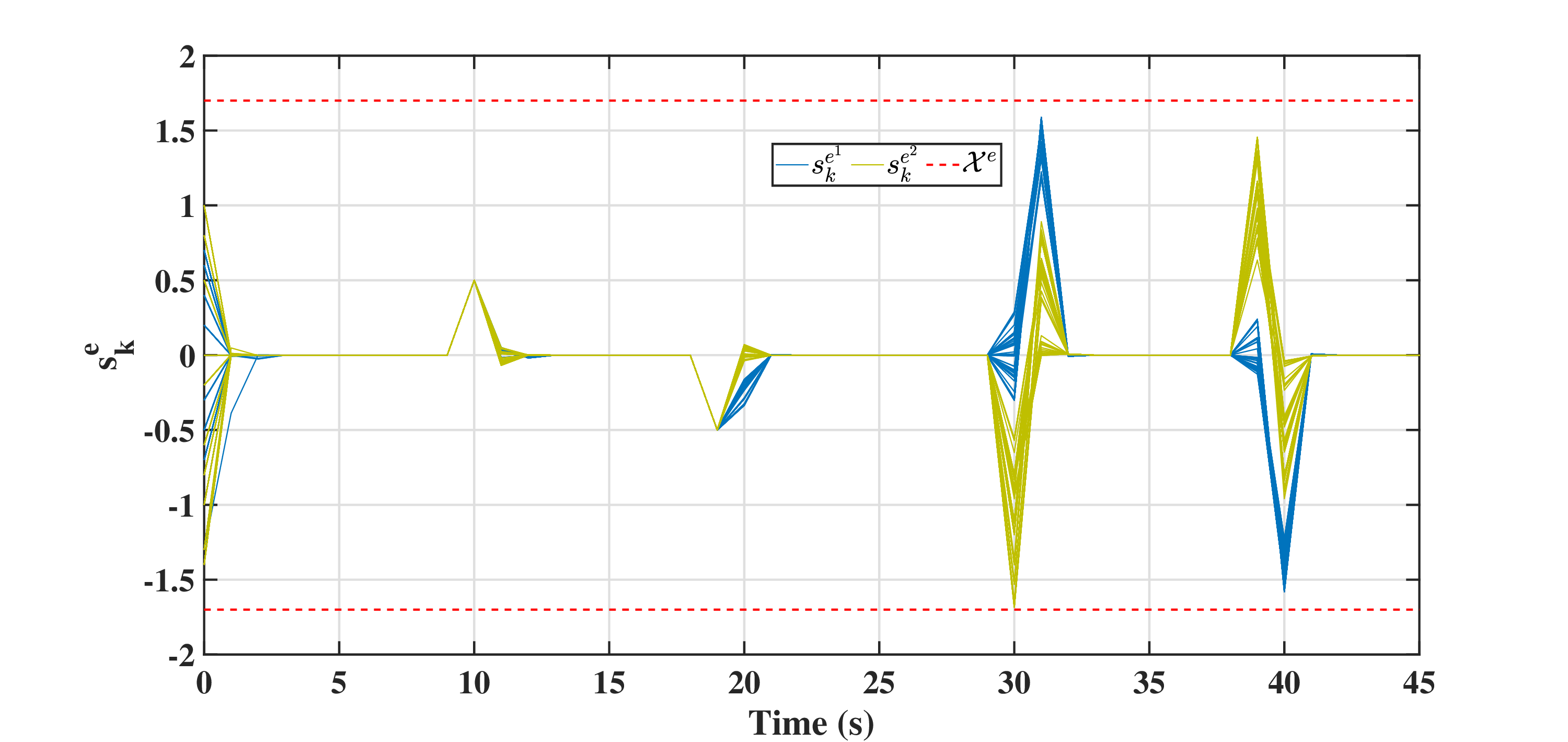}
    \caption{Error States}
    \label{fig:2}
\end{figure}
\begin{figure}
    \centering
    \includegraphics[width=\linewidth]{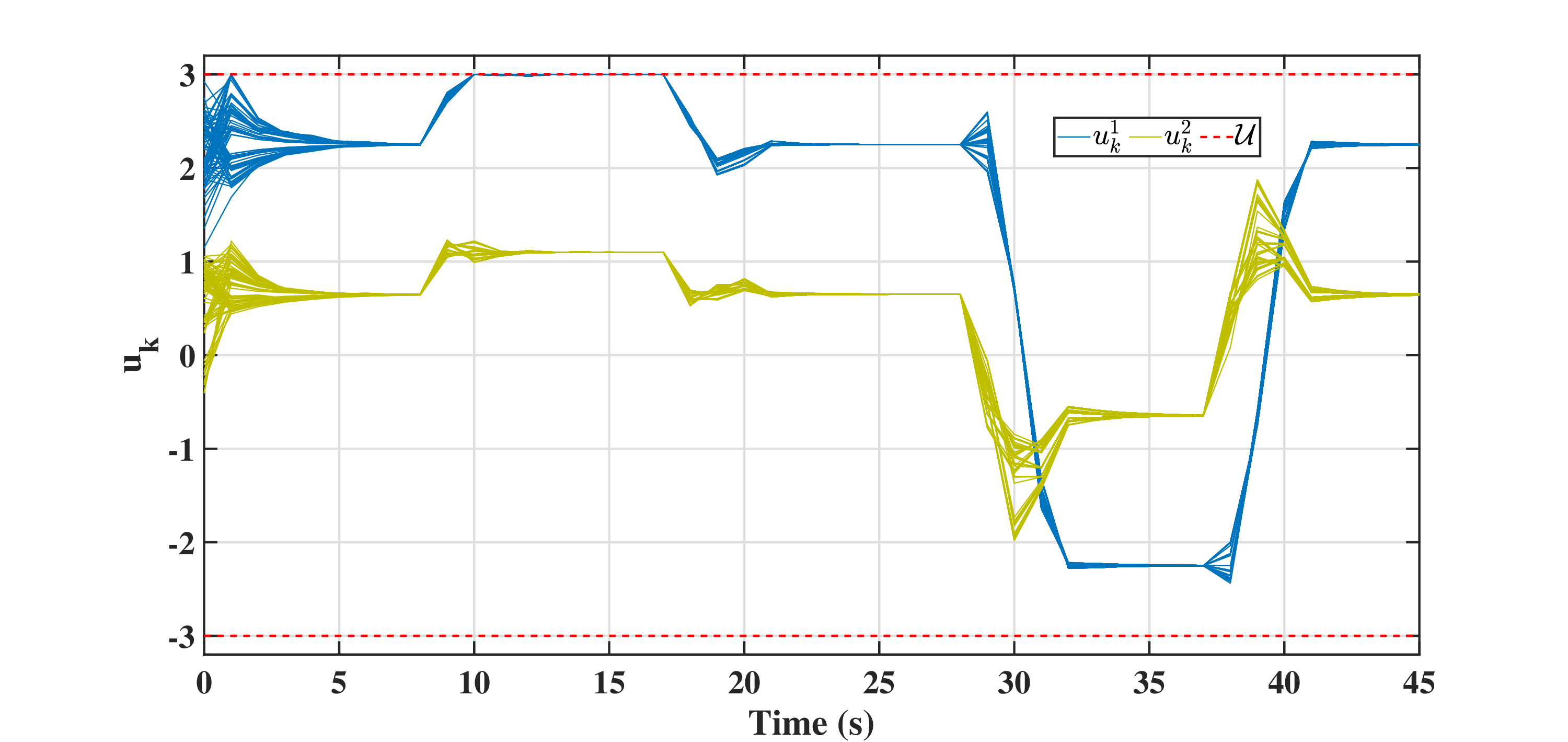}
    \caption{Control Input}
    \label{fig:3}
\end{figure}
\begin{figure}
    \centering
    \includegraphics[width=\linewidth]{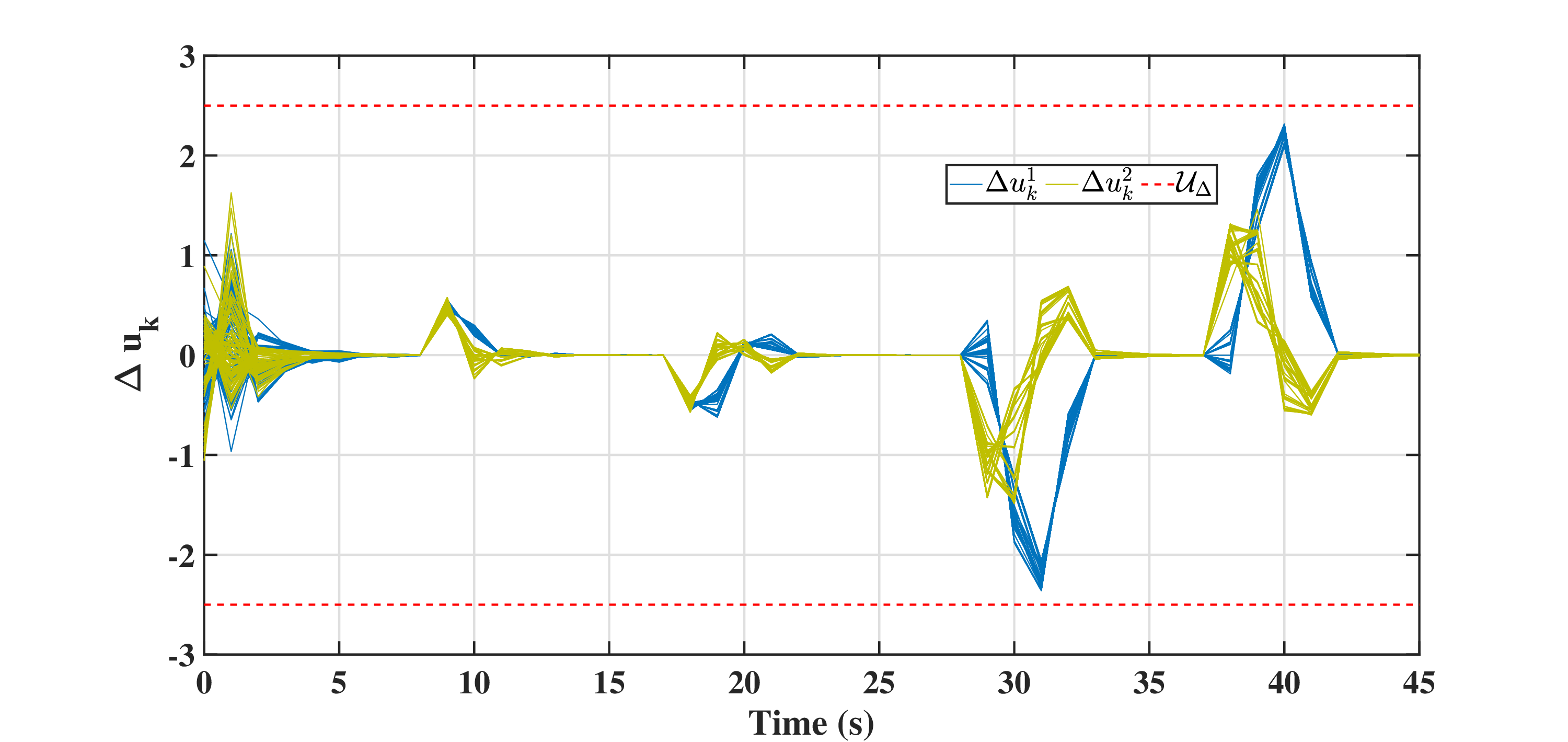}
    \caption{Rate of Change of Control Input}
    \label{fig:4}
\end{figure}
\begin{figure}
    \centering
    \includegraphics[width=\linewidth]{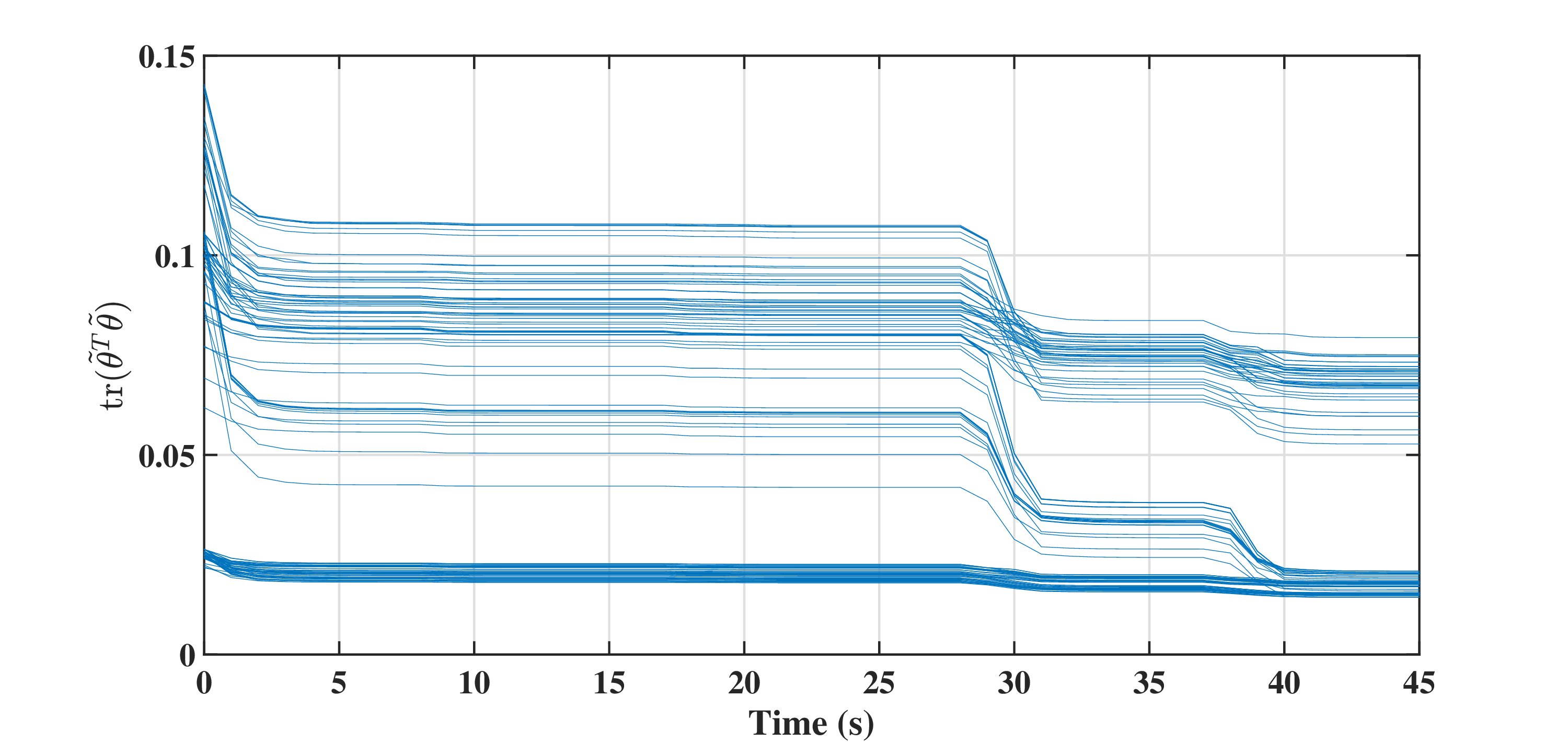}
    \caption{Parameter Estimation Error}
    \label{fig:5}
\end{figure}
\par Fig. \ref{fig:1} illustrates that the uncertain system achieves satisfactory tracking of the time-varying reference trajectory while ensuring that the system states remain within the prescribed constraint bounds. The nominal states of the error dynamics ($\bar{s}^e_k$) are also confined within their computed region as illustrated in Fig. \ref{fig:2}. The control inputs and input rate satisfy their respective constraint bounds as shown in Fig. \ref{fig:3} and Fig. \ref{fig:4}. The parameter estimation errors remain bounded as shown in Fig. \ref{fig:5}. 
In the same figure, a noticeable drop in parameter estimation errors is observed around 30s. This arises from a sudden change in the reference trajectory at that time instant. As the controller starts tracking the new reference signal, $\tilde{x}^e_{k+1}$ becomes non-zero, providing excitation to the adaptation law \eqref{eq:Adaptive}.  Thereby, resuming adaptive learning until the system converges to the new reference trajectory. 

\section{Conclusion}
An adaptive MPC control framework was proposed for the tracking of a discrete-time linear time-invariant system with bounded parametric uncertainties under hard constraints on state, input and input rate. To facilitate tracking, the tracking problem was reformulated as a regulatory problem by expressing the system dynamics in terms of their equivalent error dynamics with respect to a time-varying reference signal. Constraint reformulation in the error dynamics domain was performed to preserve the original system's physical limitations. The uncertain system parameters were adaptively estimated to construct an input to maintain zero tracking error. Despite the time-varying nature of the admissible control sequence, which rendered the standard recursive feasibility argument inapplicable, the developed solution ensured recursive feasibility by guaranteeing the continued solvability of the associated optimization problem at each time step. The tracking error dynamics were observed to converge to the origin while maintaining boundedness of the physical system states.


\bibliographystyle{unsrt}        
\bibliography{autosam}           



\end{document}